\documentclass[journal]{IEEEtran}
\IEEEoverridecommandlockouts
\usepackage{setspace}
\usepackage{cite}
\usepackage{amsmath,amssymb,amsfonts,mathtools,bm}
\usepackage{amsthm}
\usepackage{algorithm}
\usepackage{algpseudocode}
\usepackage{graphicx}
\usepackage{textcomp}
\usepackage{xcolor,float}
\usepackage{tabularx}
\usepackage{textcomp}
\usepackage{diagbox}
\usepackage{svg}
\usepackage{booktabs}
\usepackage{subfigure}
\usepackage{hyperref}
\newtheorem{theorem}{Theorem}

\usepackage{graphicx}
\usepackage{makecell}
\usepackage{amsthm}
\usepackage{caption}
\usepackage{sidecap}
\usepackage{microtype}
\usepackage{booktabs} 
\usepackage{multirow}
\usepackage{float}
\usepackage{adjustbox} 
\usepackage{amsmath}
\usepackage{amssymb}
\usepackage{colortbl}
\usepackage{makecell}
\usepackage{float}
\usepackage{url}
\usepackage{balance}
\usepackage{placeins}
\usepackage{bbding}
\usepackage{hyperref}
\usepackage{graphicx}
\usepackage{sidecap}
\usepackage{microtype}
\usepackage{graphicx}
\usepackage{subfigure}
\usepackage{booktabs} 
\usepackage{multirow}
\usepackage{array}
\usepackage{subcaption}
\usepackage{graphicx}
\usepackage{makecell}
\usepackage{amsmath}
\usepackage{amssymb}
\newtheorem{proposition}{Proposition}
\usepackage{tikz}

\definecolor{lime}{HTML}{A6CE39}
\DeclareRobustCommand{\orcidicon}{%
    \begin{tikzpicture}
    \draw[lime, fill=lime] (0,0) 
    circle [radius=0.16] 
    node[white] {{\fontfamily{qag}\selectfont \tiny ID}};    \draw[white, fill=white] (-0.0625,0.095) 
    circle [radius=0.007];    \end{tikzpicture}
    \hspace{-2mm}}
\foreach \x in {A, ..., Z}{%
    \expandafter\xdef\csname orcid\x\endcsname{\noexpand\href{https://orcid.org/\csname orcidauthor\x\endcsname}{\noexpand\orcidicon}}
    }

\allowdisplaybreaks
\renewcommand{\baselinestretch}{1}

\hypersetup{hidelinks}

\begin{document}

\title{Learn for Variation: Efficient AAV Trajectory Learning through a Differentiable Wireless World Model}

\author{
Xiucheng Wang\orcidA{},~\IEEEmembership{Graduate Student Member,~IEEE,}
Zhenye Chen\orcidB{},
Nan Cheng\orcidC{},~\IEEEmembership{Senior Member,~IEEE,} Conghao Zhou\orcidF{},~\IEEEmembership{Member,~IEEE,}
Zhisheng Yin\orcidH{},~\IEEEmembership{Member,~IEEE,}
Xuemin (Sherman) Shen\orcidG{},~\IEEEmembership{Fellow,~IEEE}

\thanks{
\par This work was supported by the National Key Research and Development Program of China (2024YFB907500).
\par Xiucheng Wang, Nan Cheng, Conghao Zhou, and Zhisheng Yin are with the State Key Laboratory of ISN and School of Telecommunications Engineering, Xidian University, Xi’an 710071, China (e-mail: xcwang\_1@stu.xidian.edu.cn; dr.nan.cheng@ieee.org, conghao.zhou@ieee.org, zsyin@xidian.edu.cn);\textit{(Corresponding author: Nan Cheng.)}.
\par Zhenye Chen is with the School of Aerospace Science and Technology, Xidian University, Xi’an 710071, China (e-mail: 24149100028@stu.xidian.edu.cn)
\par Xuemin (Sherman) Shen is with the Department of Electrical and Computer Engineering, University of Waterloo, Waterloo, N2L 3G1, Canada (e-mail: sshen@uwaterloo.ca).

}

} 
    \maketitle

\IEEEdisplaynontitleabstractindextext

\IEEEpeerreviewmaketitle

\begin{abstract}
Autonomous aerial vehicles (AAVs) enable data collection for sixth-generation Internet-of-Things networks, but their trajectories couple nonlinear wireless rates with long-horizon service progress. This paper views the evolution of AAV kinematics, channel state, and user backlog as a structured differentiable world model and develops Learn for Variation (L4V) to exploit that model efficiently. L4V replaces a discontinuous completion-time objective with a cumulative-backlog surrogate, unrolls the mission dynamics, and propagates pathwise sensitivities to a neural policy through the discrete adjoint recursion. The resulting derivative is exact conditional on a fixed exogenous-noise realization; stochastic expected-objective optimization still requires sampling. We show that the structured adjoint grows at most polynomially with the horizon and establish a stationary-point rate for fixed-step full-gradient descent under standard smoothness assumptions. The framework also learns shared OFDMA allocation under reparameterized shadowing and Rician fading, while distributional pretraining amortizes model-based optimization into forward-only deployment on unseen layouts. Paired stress tests cover channel-generator mismatch, noisy partial observations, a fixed-resource two-AAV extension, and a circular no-fly region. Code and configurations are available at \url{https://github.com/UNIC-Lab/L4V-AAV}. Against genetic-algorithm, DQN, A2C, DDPG, and differentiable model-predictive-control implementations, L4V reduces mission time by up to $65\%$, executes a default mission in $53$ ms, and completes all $60$ frozen-policy tests after pretraining on $1{,}600$ layouts.
\end{abstract}

\begin{IEEEkeywords}
Autonomous aerial vehicles, trajectory optimization, structured world models, differentiable simulation, adjoint gradients, zero-shot generalization.
\end{IEEEkeywords}

\section{Introduction}
Autonomous aerial vehicles (AAVs) can provide reconfigurable connectivity for sixth-generation (6G) Internet-of-Things (IoT) data collection \cite{shen2023toward,wang2022joint}. Their motion changes the air-to-ground geometry and hence the service delivered to spatially separated users \cite{10764739,zhang2022challenges}. Communication-aware trajectory planning therefore selects a sequence of positions, rather than a single favorable location, to complete all user demands \cite{ding20203d,zeng2017energy}.

Three features make this sequential problem nontrivial. First, the rate changes nonlinearly with position \cite{zhou2022two}. Second, heterogeneous unfinished demands and a shared radio-resource budget couple the users at every slot. Third, mission completion is a first-hitting-time objective: an early move affects both current service and the geometry from which later users are approached \cite{zeng2016wireless,liu2019distributed}. A terminal completion-time reward can therefore provide a coarse learning signal, whereas differentiating the remaining backlog at every slot exposes how each action changes subsequent service.

Model-based trajectory design provides explicit constraints and interpretable optimization variables, but often requires an iterative solve \cite{liu2018energy,liu2019distributed,ding20203d}. Deep reinforcement learning (DRL) does not require environment derivatives, but can require many interactions and careful reward shaping \cite{wang2022joint,10623528}. World-model methods instead use a predictive dynamics model for planning or policy learning \cite{ha2018world,hafner2020dream}. In this paper, ``world model'' refers specifically to the structured transition that maps AAV motion, stochastic radio propagation, resource allocation, and residual demand to the next communication state; the model is physics- and communications-informed rather than learned from pixels. L4V exploits its differentiability to obtain pathwise derivatives \cite{heiden2021neuralsim}. The broad sweeps below optimize each learned policy for its sampled layout to compare optimization quality; a separate distributional-pretraining experiment freezes one policy and evaluates it without updates on disjoint held-out layouts.

To bridge this gap, this paper develops Learn for Variation (L4V), a gradient-informed trajectory learning framework that efficiently turns a structured communication world model into policy supervision. The key idea is to replace the discontinuous completion-time objective with a cumulative-backlog surrogate and to unroll the entire flight evolution as a differentiable computation graph. Backpropagation through time then evaluates the corresponding discrete adjoint recursion \cite{ioffe1997euler,kim2010optimal} and propagates pathwise sensitivities from the mission objective to the policy parameters. Conditional on a fixed channel-noise realization, these derivatives are exact and avoid policy-action sampling. Under stochastic training, they remain pathwise gradient samples rather than zero-variance gradients of the expected objective. L4V further incorporates temporal smoothness regularization and norm-based gradient clipping to improve trajectory regularity and numerical stability, and distributional pretraining amortizes repeated model-based optimization into a frozen forward policy. The main contributions of this paper are summarized as follows.
\begin{enumerate}
\item We formulate communication-aware AAV data collection as a structured world model whose explicit state transition jointly represents AAV kinematics, stochastic radio propagation, shared OFDMA allocation, and residual user demand. A cumulative-backlog objective converts mission completion into dense long-horizon supervision without discarding the physical structure.

\item We derive first-order necessary conditions through the Pontryagin minimum principle \cite{kim2010optimal} and show that backpropagation through the world model implements the same discrete adjoint recursion. This efficiently exposes how every motion and allocation action changes future service, yielding exact conditional pathwise derivatives for each fixed realization of the exogenous noise.

\item We propose L4V, which integrates the adjoint gradients into deterministic neural policy learning. Distributional pretraining amortizes this model-based optimization into forward-only zero-shot execution on disjoint layouts; temporal smoothness regularization and norm-based gradient clipping improve long-horizon training stability.

\item We conduct systematic simulations against representative baselines spanning evolutionary optimization, value-based reinforcement learning, actor-critic reinforcement learning, and gradient-based optimal control. In addition to broad parameter sweeps, paired stress tests evaluate channel-generator mismatch, noisy and partially observed state, a fixed-resource two-AAV extension, circular no-fly regions, the heading-smoothness coefficient, and frozen-policy transfer after large-scale layout pretraining.
\end{enumerate}

The remainder of this paper is organized as follows. Section~\ref{sec-2} reviews related work. Section~\ref{sec-3} presents the system model and problem formulation. Section~\ref{sec-4} develops the differentiable sequence optimization method and the associated theoretical analysis. Section~\ref{sec-5} details the L4V framework and its training procedure. Section~\ref{sec-6} reports simulation results and comparative evaluations. Section~\ref{sec-7} concludes the paper and discusses future directions.

\section{Related Work}\label{sec-2}
The deployment of AAVs in wireless networks introduces trajectory-design problems with nonlinear communication objectives and long decision horizons. Existing studies mainly use mathematical programming, graph-based planning, or reinforcement learning.

The dominant framework for offline trajectory planning relies on mathematical optimization. Because trajectory, transmit power, and user scheduling are coupled within the rate expression, block coordinate descent decomposes the joint problem into subproblems \cite{tseng2001convergence}. Successive convex approximation handles non-convex distance terms through local surrogates \cite{liu2019stochastic} and can support fairness objectives through epigraph reformulations \cite{yang2019inexact}. These methods offer interpretable constraints and convergence properties, but repeatedly solving quadratically constrained programs can impose substantial online latency \cite{5447068}. Their analytical models also require explicit approximations for obstacles and multipath.

A second line discretizes the AAV mission into graph decisions. Hover locations obtained from clustering or communication-disk intersections \cite{el2020improved} reduce planning to travelling-salesman \cite{xu2019brief} or vehicle-routing \cite{rojas2021unmanned} variants. Grid-based search can additionally enforce collision avoidance. This abstraction supports global routing, but it often suppresses physical-layer dynamics \cite{erdos2013experimental} and therefore requires a separate refinement stage for communication-aware motion control.

Deep reinforcement learning (DRL) instead frames trajectory design as a Markov decision process \cite{wang2022joint}. Algorithms such as DDPG \cite{silver2014deterministic} and twin delayed DDPG \cite{dankwa2019twin} learn continuous controls, while proximal policy optimization is frequently used in multi-agent settings \cite{schulman2017proximal}. Model-free training does not require environment derivatives, but it can require many interactions and careful reward shaping. Long missions further increase the variance of sampled policy-gradient estimates and complicate temporal credit assignment.

Learned world models compress observations into predictive latent dynamics and can support control through imagined rollouts \cite{ha2018world,hafner2020dream}. Differentiable simulators likewise expose model derivatives to gradient-based control and policy learning \cite{heiden2021neuralsim}. L4V therefore does not claim either the world-model concept or backpropagation through dynamics as a new algorithm. Its distinction is to make the wireless channel part of an explicit communication world model and connect shared-resource service dynamics, a cumulative-backlog surrogate, a discrete adjoint analysis, and a deployable neural trajectory policy. Unlike latent visual world models, L4V does not learn the transition law from raw observations; it assumes a known differentiable generator and tests mismatch directly. Distributional pretraining can move the resulting optimization cost offline; Section~\ref{sec-6j} evaluates the frozen policy on layouts disjoint from pretraining instead of inferring transfer from per-layout experiments.

\section{System Model and Problem Formulation}\label{sec-3}
We consider a wireless communication system where a single AAV collects data from $N$ ground users, $\mathcal{N}=\{1,\ldots,N\}$, at fixed horizontal coordinates $\mathbf{w}_i=[x_i,y_i]^T$. User $i$ has an initial normalized demand $D_i$. The mission is discretized into $T$ slots of duration $\tau$. Consistent with the released simulator, $v[t]\in[0,v_{\max}]$ denotes normalized horizontal displacement per slot; hence $\mathbf{q}[t+1]=\mathbf{q}[t]+v[t][\cos\theta[t],\sin\theta[t]]^T$. Physical speed can be recovered through a coordinate and time conversion when evaluating propulsion energy. The control $\mathbf{u}[t]=[v[t],\theta[t],\mathbf{a}[t]^T]^T$ also includes a resource-allocation vector $\mathbf{a}[t]=[a_1[t],\ldots,a_N[t]]^T$. A global softmax enforces $\sum_i a_i[t]=1$ and $a_i[t]\geq0$, so users compete for one fixed fleet-wide resource budget. The normalized service rate is
\begin{align}
    R_i[t] = N a_i[t]\ln\!\left(1+
    \frac{\eta\,\xi_i[t]}
    {\sigma^2\!\left[\max\!\left\{\|\mathbf{q}[t]-\mathbf{w}_i\|^2,\delta\right\}\right]^{\nu/2}}\right),
    \label{eq:rate_extended}
\end{align}
where $\eta$ is the unit-distance gain, $\sigma^2$ is normalized noise power, $\nu$ is the path-loss exponent, $\delta=0.1$ prevents the singularity at zero horizontal distance, and $\xi_i[t]\geq0$ is stochastic fading. The factor $N$ keeps the total resource fixed while making $a_i[t]=1/N$ recover the original per-user normalized rate exactly. Thus the model is an algorithmic test bed rather than a calibrated physical-layer link budget. The broad baseline sweeps use $a_i[t]=1/N$, $\xi_i[t]=1$, and $\nu=2$; the extended experiments use learned allocation and the fading term
\begin{align}
    \xi_i[t] = \exp\!\big(\kappa \epsilon_i[t]\big) \cdot \left| \sqrt{\tfrac{K_{\mathrm{r}}}{K_{\mathrm{r}}+1}} + \sqrt{\tfrac{1}{2(K_{\mathrm{r}}+1)}}\big(\epsilon_i^{\mathrm{r}}[t] + j\epsilon_i^{\mathrm{i}}[t]\big) \right|^2,
    \label{eq:fading}
\end{align}
where $\epsilon_i[t],\epsilon_i^{\mathrm r}[t],\epsilon_i^{\mathrm i}[t]\sim\mathcal{N}(0,1)$ are independent across users and slots, $\kappa\geq0$ is the log-shadowing standard deviation, and $K_{\mathrm r}\geq0$ is the Rician K-factor. Taking $K_{\mathrm r}\to\infty$ removes small-scale fading; the channel is deterministic only when $\kappa=0$ also disables shadowing. For a fixed draw, the realized rollout has an exact pathwise derivative; reparameterization preserves differentiability but does not remove uncertainty. Training assumes known kinematics and a known differentiable channel generator. Section~\ref{sec-6j} separately measures degradation when the deployment generator or state observations differ from training. The core derivation still assumes fixed users, continuous controls, and a differentiable rate; the multi-AAV and obstacle tests are structural extensions, not a claim of hardware realism. Let $T_i$ denote the slot at which user $i$ finishes. The optimization problem is
\begin{subequations}
\begin{align}
    (\text{P1}): \quad \min_{ \{\mathbf{u}[t]\} } \quad & \frac{1}{N} \sum_{i=1}^{N} T_i \\
    \text{s.t.} \quad & \sum_{t=1}^{T_i} R_i[t] \tau \ge D_i, \quad \forall i \in \mathcal{N} \\
    & 0 \le v[t] \le v_{\max}, \quad \forall t \\
    & \theta[t] \in [0, 2\pi), \quad \forall t \\
    & \sum_{i=1}^{N} a_i[t] = 1, \quad a_i[t] \ge 0, \quad \forall t \label{eq:p1_alloc}
\end{align}
\end{subequations}
where the trajectory $\mathbf{q}[t]$ embedded in $R_i[t]$ is determined by the control sequence $\{\mathbf{u}[t]\}$. Problem (P1) is nonconvex: position affects every user's rate over time, and the simplex allocation couples users through one time-varying resource budget. We therefore seek a policy that uses the known transition structure without solving a new long-horizon program at every deployment.

\section{World-Model-Guided Differential Sequence Optimization}\label{sec-4}
\begin{figure*}[!t]
    \centering
    \includegraphics[trim=0 0 0 25pt,clip,width=1\linewidth]{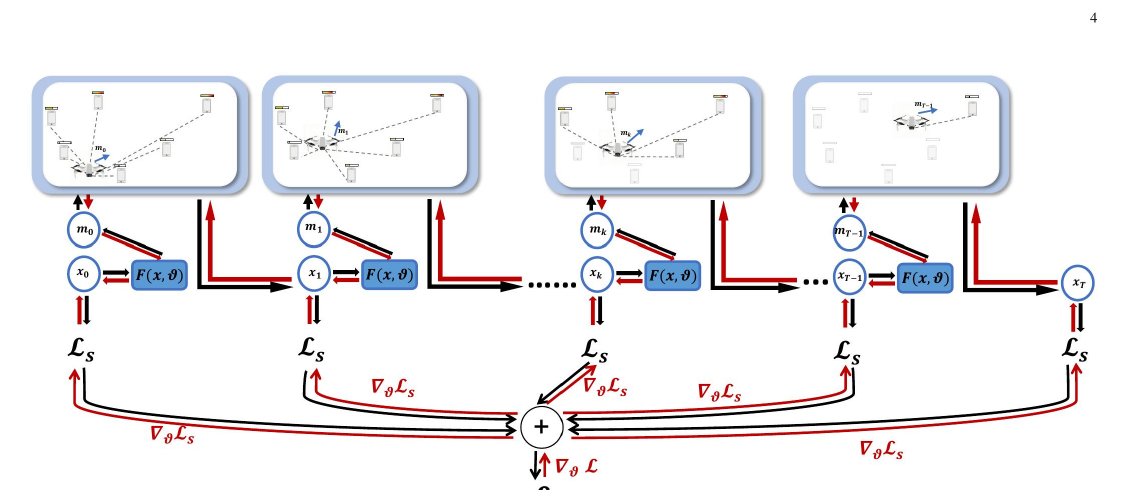}
    \caption{The L4V framework. Each blue transition block is one step of the structured wireless world model. The forward pass (black) unrolls the AAV trajectory and accumulates $\mathcal{L}$; the backward pass (red) implements the discrete adjoint and propagates exact conditional pathwise gradients $\nabla_\theta \mathcal{L}$ to the policy parameters.}
    \label{fig-system}
\end{figure*}

\subsection{Theoretical Analysis}
The transition equations in Section~\ref{sec-3} define a structured communication world model: given the current AAV location, residual demands, resource allocation, and reparameterized channel noise, they predict the next location and service state. Directly optimizing (P1) requires a long-horizon nonconvex solve. L4V instead unrolls this model, as illustrated in Fig.~\ref{fig-system}, and uses its derivatives through the adjoint method. To expose the structure of the resulting supervision, we first analyze a localized version of the problem.
Average completion time can be written as the area under the number of active users,
\begin{align}
    \frac{1}{N}\sum_{i=1}^{N}T_i
    =\frac{1}{N}\sum_{t=0}^{T-1}\sum_{i=1}^{N}\mathbf{1}\{d_i[t]>0\},
    \label{eq:completion_indicator}
\end{align}
provided the horizon covers every completion event. This exact expression is discontinuous. Defining the active set as $\mathcal{A}[t]=\{i\in\mathcal{N}:d_i[t]>0\}$, a natural greedy heuristic instead chooses $\mathbf{u}[t]$ to maximize the aggregate rate of the currently active users,
\begin{align}
    (\text{P2}): \quad \max_{\mathbf{u}[t]} \quad J(\mathbf{u}[t]) = \sum_{i \in \mathcal{A}[t]} R_i[t+1],
\end{align}
where $R_i[t+1]$ is the rate at the next position $\mathbf{q}[t+1]$, a direct function of $\mathbf{u}[t]$.
This greedy policy is myopic: a move that maximizes next-slot service can leave the AAV poorly positioned for later users. We instead optimize the full state sequence. The state is $\mathbf{x}[t]=[\mathbf{q}[t]^T,\mathbf{d}[t]^T]^T\in\mathbb{R}^{2+N}$, where $\mathbf{d}[t]=[d_1[t],\ldots,d_N[t]]^T$ is the remaining-demand vector. The control $\mathbf{u}[t]=[v[t],\theta[t],\mathbf{a}[t]^T]^T\in\mathbb{R}^{2+N}$ contains the per-slot displacement, heading, and simplex-constrained resource allocation. The transition $\mathbf{x}[t+1]=\mathbf{f}(\mathbf{x}[t],\mathbf{u}[t])$ consists of
\begin{align}
\mathbf{q}[t+1]&=\mathbf{q}[t]+v[t][\cos\theta[t],\sin\theta[t]]^T,\\
d_i[t+1]&=\max\{0,d_i[t]-R_i[t]\tau\}.
\end{align}
We use the stage cost $L(\mathbf{x}[t])=\sum_i d_i[t]$ and minimize its horizon sum:
\begin{align}
    (\text{P3}): \quad \min_{\{\mathbf{u}[t]\}_{t=0}^{T-1}} \quad & J = \sum_{t=0}^{T-1} L(\mathbf{x}[t]) = \sum_{t=0}^{T-1} \sum_{i=1}^{N} d_i[t] \\
    \text{s.t.} \quad & \mathbf{x}[t+1] = \mathbf{f}(\mathbf{x}[t], \mathbf{u}[t]), \quad \forall t \in [0, T-1] \nonumber \\
    & 0 \le v[t] \le v_{\max}, \quad \theta[t] \in [0, 2\pi), \quad \forall t \nonumber \\
    & \mathbf{a}[t]\succeq\mathbf{0},\quad \mathbf{1}^{T}\mathbf{a}[t]=1,\quad \forall t \nonumber\\
    & \mathbf{x}[0] = \mathbf{x}_0 = [\mathbf{q}_0^T, D_1, \dots, D_N]^T . \nonumber
\end{align}
\noindent Problem (P3) is a cumulative-backlog surrogate for (P1), not an equivalent reformulation. Equation~\eqref{eq:completion_indicator} weights every unfinished user equally, whereas (P3) weights it by the magnitude of its remaining data. The surrogate supplies a dense per-step signal and encourages early service, but two policies can have different backlog areas even when their completion times coincide, and minimizing (P3) does not guarantee that all users finish within a finite $T$. In the experiments, the horizon is capped at $500$ slots and success is declared when the total remaining task falls below $10^{-3}N$; capped runs are retained rather than silently discarded. This distinction is important when interpreting completion-time results.
For an independently parameterized control sequence, first-order necessary conditions for (P3) follow from the discrete-time Pontryagin Minimum Principle (PMP). Introducing the Hamiltonian makes the action sensitivities used below explicit:
\begin{align}
    \mathcal{H}(\mathbf{x}[t], \mathbf{u}[t], \boldsymbol{\lambda}[t+1]) = L(\mathbf{x}[t], \mathbf{u}[t]) + \boldsymbol{\lambda}[t+1]^T \mathbf{f}(\mathbf{x}[t], \mathbf{u}[t])
\end{align}
where $\boldsymbol{\lambda}[t] \in \mathbb{R}^{2+N}$ is the co-state (or adjoint) vector. It equals the sensitivity of the total cost to a state perturbation, $\boldsymbol{\lambda}[t] = \nabla_{\mathbf{x}[t]} J$, and follows the backward recursion
\begin{align}
    \boldsymbol{\lambda}[t] = \nabla_{\mathbf{x}[t]} \mathcal{H}(\mathbf{x}[t], \mathbf{u}[t], \boldsymbol{\lambda}[t+1])
\end{align}
with terminal condition $\boldsymbol{\lambda}[T]=\mathbf{0}$ because (P3) has no terminal cost. A necessary condition for an independently optimized control sequence $\{\mathbf{u}^*[t]\}$ is Hamiltonian minimization at every slot. In the interior of the control domain, this implies
\begin{align}
    \nabla_{\mathbf{u}[t]} \mathcal{H}(\mathbf{x}^*[t], \mathbf{u}^*[t], \boldsymbol{\lambda}^*[t+1]) = \mathbf{0}
    \label{eq:pmp_control}
\end{align}
When the displacement constraint is active, its component obeys the corresponding boundary condition, while the heading component remains unconstrained. The allocation component instead obeys the Karush--Kuhn--Tucker conditions for \eqref{eq:p1_alloc}: $\nabla_{a_i[t]}\mathcal{H}-\mu[t]\geq0$, with equality whenever $a_i^*[t]>0$. The softmax parameterization in Section~\ref{sec-5} enforces simplex feasibility during training. The Hamiltonian derivative gives the total sensitivity to an independently perturbed action and thus supplies structured, model-based temporal credit assignment. Under a stochastic channel, it is deterministic only conditional on the sampled exogenous-noise sequence; compared with a score-function estimator, it avoids the additional variance caused by sampling policy actions.
Inspired by this, we design a deterministic neural policy, $\pi_{\boldsymbol{\theta}}$, that directly parameterizes the control law as a function of the current state:
\begin{align}
    \mathbf{u}[t] = \pi_{\boldsymbol{\theta}}(\mathbf{x}[t])
\end{align}
Here, $\boldsymbol{\theta}$ represents the trainable policy parameters. The network outputs $\mathbf{u}[t]=[v[t],\theta[t],\mathbf{a}[t]^T]^T$: the displacement head passes through a sigmoid scaled by $v_{\max}$, the heading head is bounded to $[-\pi,\pi]$, and the allocation logits pass through a global softmax. Because the state fed back to the policy depends on $\boldsymbol{\theta}$, the parameter gradient must use the closed-loop rather than the open-loop adjoint. Define
\begin{align}
    \mathbf{F}_{\boldsymbol{\theta}}(\mathbf{x})
    &\triangleq \mathbf{f}(\mathbf{x},\pi_{\boldsymbol{\theta}}(\mathbf{x})),&
    \ell_{\boldsymbol{\theta}}(\mathbf{x})
    &\triangleq L(\mathbf{x},\pi_{\boldsymbol{\theta}}(\mathbf{x})).
\end{align}
For $J=\sum_{t=0}^{T-1}\ell_{\boldsymbol{\theta}}(\mathbf{x}[t])$, the closed-loop adjoint satisfies
\begin{align}
    \boldsymbol{\lambda}[T]&=\mathbf{0},\\
    \boldsymbol{\lambda}[t]
    &=\nabla_{\mathbf{x}}\ell_{\boldsymbol{\theta}}(\mathbf{x}[t])
      +
      \left(\nabla_{\mathbf{x}}\mathbf{F}_{\boldsymbol{\theta}}(\mathbf{x}[t])\right)^T
      \boldsymbol{\lambda}[t+1],
    \label{eq:closed_loop_adjoint}
\end{align}
and
\begin{align}
    \nabla_{\boldsymbol{\theta}}J
    =\sum_{t=0}^{T-1}\!\left[
      \nabla_{\boldsymbol{\theta}}\ell_{\boldsymbol{\theta}}(\mathbf{x}[t])
      +\left(\nabla_{\boldsymbol{\theta}}\mathbf{F}_{\boldsymbol{\theta}}(\mathbf{x}[t])\right)^T
      \boldsymbol{\lambda}[t+1]\right].
    \label{eq:closed_loop_gradient}
\end{align}
Reverse-mode AD evaluates \eqref{eq:closed_loop_adjoint}--\eqref{eq:closed_loop_gradient} without explicitly forming the Jacobians. The open-loop Hamiltonian derivative $\nabla_{\mathbf{u}[t]}\mathcal{H}$ remains the action-level sensitivity, while the closed-loop recursion additionally accounts for the policy-state Jacobian and parameter sharing across time. Thus parameter-space stationarity does not generally imply the action-wise PMP condition at every slot.

\subsection{Environment with Differentiable Transition}
The structured wireless world model is implemented as one differentiable computational graph. For a fixed initial state and exogenous-noise draw, the time-unrolled transition and stage cost define a scalar function $J=\mathcal{G}(\mathbf{x}[0],\boldsymbol{\theta})$. Reverse-mode automatic differentiation (AD) evaluates the closed-loop adjoint in \eqref{eq:closed_loop_adjoint}--\eqref{eq:closed_loop_gradient} without explicitly materializing horizon-wide Jacobians.

The transition and stage cost use differentiable elementary operations except for backlog clipping. The kinematic update uses addition and trigonometric functions; the task update uses $\max\{0,\cdot\}$, which is differentiable away from its kink and receives the implementation-selected subgradient at the kink. Equation~\eqref{eq:rate_extended} combines the clamped distance, logarithm, softmax allocation, and reparameterized fading. We define the local transition Jacobians
    \begin{align}
        \mathbf{A}[t] \triangleq \frac{\partial \mathbf{f}(\mathbf{x}[t], \mathbf{u}[t])}{\partial \mathbf{x}[t]} \in \mathbb{R}^{(2+N) \times (2+N)} \\
        \mathbf{B}[t] \triangleq \frac{\partial \mathbf{f}(\mathbf{x}[t], \mathbf{u}[t])}{\partial \mathbf{u}[t]} \in \mathbb{R}^{(2+N) \times  (2+N)}
    \end{align}
These matrices contain the linearized dynamics around the current trajectory point. The instantaneous cost, $L(\mathbf{x}[t]) = \sum_{i=1}^{N} d_i[t]$, is a linear sum, which is trivially differentiable. Its partial derivatives are:
    \begin{align}
        &\frac{\partial L(\mathbf{x}[t])}{\partial \mathbf{x}[t]} = \begin{bmatrix} \mathbf{0}_{2 \times 1} \\ \mathbf{1}_{N \times 1} \end{bmatrix}^T, \\
        &\frac{\partial L(\mathbf{x}[t])}{\partial \mathbf{u}[t]} = \mathbf{0}_{1 \times  (2+N)}
    \end{align}

During a rollout, the policy and transition alternate for $T$ slots and accumulate $J=\sum_{t=0}^{T-1}L(\mathbf{x}[t])$. Applying reverse-mode AD to this scalar traverses the same graph backward: its state sensitivity is the co-state, its action sensitivity is the Hamiltonian derivative in \eqref{eq:pmp_control}, and parameter sharing is handled by the closed-loop recursion. We state these identities once rather than duplicating the adjoint and action-gradient equations in the implementation section. For a fixed exogenous-noise realization, the resulting pathwise sensitivities are exact; variance across stochastic channel realizations is not eliminated.

\section{Learning to Exploit the Wireless World Model}\label{sec-5}
The L4V framework uses two interconnected modules: a parameterized neural policy and the structured differentiable wireless world model. The policy $\pi_{\boldsymbol{\theta}}$ maps the current communication state to $\mathbf{u}[t] = \pi_{\boldsymbol{\theta}}(\mathbf{x}[t])$. The world model $\mathcal{S}$ embodies the kinematic, channel, resource, backlog, and cost equations. During training, the two modules are unrolled in closed loop over horizon $T$, so that $\mathbf{x}[t+1] = \mathbf{f}(\mathbf{x}[t], \pi_{\boldsymbol{\theta}}(\mathbf{x}[t]))$ and $\mathbf{x}[t] = \Phi_t(\mathbf{x}, \boldsymbol{\theta})$. Its explicit differentiable structure exposes the Jacobians $\frac{\partial \mathbf{f}}{\partial \mathbf{x}[t]}$ and $\frac{\partial \mathbf{f}}{\partial \mathbf{u}[t]}$ rather than treating the transition as a black box. Consequently, the total cost is a single differentiable function of the policy parameters:
\begin{align}
    J(\boldsymbol{\theta}) = \sum_{t=0}^{T-1} L(\mathbf{x}[t]) = \sum_{t=0}^{T-1} L(\Phi_{t}(\mathbf{x}, \boldsymbol{\theta}))
\end{align}
This formulation expresses the trajectory problem as finite-dimensional policy-parameter optimization. It permits reverse-mode differentiation of the unrolled loss, but the resulting objective remains nonconvex.

\subsection{Training Procedure}
The training procedure follows a ``simulation as computation'' paradigm and is deterministic conditional on the initial state and exogenous-noise draw. Each iteration performs a forward rollout and a reverse pass for the conditional pathwise gradient. Starting with an initial state $\mathbf{x}$, the trajectory is generated recursively for $t\in[0,T-1]$:
\begin{align}
    \mathbf{u}[t] &= \pi_{\boldsymbol{\theta}}(\mathbf{x}[t]) \label{eq:policy_fwd} \\
    \mathbf{x}[t+1] &= \mathbf{f}(\mathbf{x}[t], \mathbf{u}[t]) \label{eq:dynamics_fwd}
\end{align}
This forward pass is executed within an automatic differentiation (AD) framework, which records every computation and its dependencies in a dynamic computational graph. Concurrently, the total cost functional $J$ is accumulated by summing the instantaneous costs incurred at each step:
\begin{align}
    J(\boldsymbol{\theta}) = \sum_{t=0}^{T-1} L(\mathbf{x}[t])
\end{align}
Upon completion of the rollout, $J(\boldsymbol{\theta})$ exists as the terminal node of the computational graph, representing the total cost as a differentiable function of the policy parameters $\boldsymbol{\theta}$.

The core of the L4V methodology lies in the implementation of the backward pass. By invoking a reverse-mode AD routine on the scalar cost $J$, we compute the total derivative $\frac{dJ}{d\boldsymbol{\theta}}$. This operation, commonly known as Backpropagation Through Time (BPTT), implicitly performs the discrete adjoint method, providing an efficient and exact calculation of the required variational gradients. The total gradient is found by applying the chain rule over the entire unrolled graph, which decomposes into a sum over the time steps:
\begin{align}
    \frac{dJ}{d\boldsymbol{\theta}} = \sum_{t=0}^{T-1} \frac{dJ}{d\mathbf{u}[t]} \frac{\partial \mathbf{u}[t]}{\partial \boldsymbol{\theta}}
\end{align}
Here $\frac{\partial \mathbf{u}[t]}{\partial \boldsymbol{\theta}} = \nabla_{\boldsymbol{\theta}}\pi_{\boldsymbol{\theta}}(\mathbf{x}[t])$ is the policy Jacobian at fixed state. BPTT obtains $\frac{dJ}{d\mathbf{u}[t]}$ as the Hamiltonian action sensitivity in \eqref{eq:pmp_control}, while the closed-loop recursion \eqref{eq:closed_loop_adjoint} accounts for the state dependence of later policy outputs. Each iteration therefore returns the exact conditional pathwise parameter gradient and applies an optimizer step; it does not assert that parameter stationarity is equivalent to slot-wise Hamiltonian stationarity.

The implementation adds heading regularization and gradient clipping. In exact agreement with the released training code, the smoothness term is the mean squared change of the bounded heading output:
\begin{align}
    J_{\text{smooth}}(\boldsymbol{\theta}) =
    \frac{1}{T-1}\sum_{t=1}^{T-1}
    \big(\theta[t]-\theta[t-1]\big)^2.
    \label{eq:smooth}
\end{align}
It discourages slot-to-slot steering changes but is not a full actuator or acceleration model. Section~\ref{sec-6j} reports sensitivity to its coefficient $\beta$.

Beyond control smoothness, flight endurance is a first-order concern for battery-limited AAV platforms, so we optionally augment the objective with a flight-energy term based on the closed-form rotary-wing propulsion power model of Zeng and Zhang~\cite{zeng2017energy}:
\begin{align}
    P(\bar v[t]) = P_0\!\left(1+\frac{3\bar v[t]^2}{U_{\mathrm{tip}}^2}\right) + P_i\!\left(\sqrt{1+\frac{\bar v[t]^4}{4v_0^4}} - \frac{\bar v[t]^2}{2v_0^2}\right)^{\!1/2} \notag \\
    {}+ \frac{1}{2}d_0\rho s A\, \bar v[t]^3,
    \label{eq:energy}
\end{align}
where $\bar v[t]=c_vv[t]$ converts normalized displacement to the speed used by the propulsion model; the reported normalized experiment uses $c_v=1$. The remaining constants have their standard rotary-wing meanings \cite{zeng2017energy}. Accumulating the differentiable power term gives
\begin{align}
    J_{\text{energy}}(\boldsymbol{\theta}) = \sum_{t=0}^{T-1} P(\bar v[t])\, \tau,
    \label{eq:j_energy}
\end{align}
which Section~\ref{sec-6} uses to trade off completion time against flight energy as its weight $\gamma$ varies.

The new, regularized total cost functional, $J_{\text{total}}(\boldsymbol{\theta})$, is a weighted sum of the original task-oriented cost, the smoothness penalty, and, when enabled, the energy cost:
\begin{align}
    J_{\text{total}}(\boldsymbol{\theta}) = J(\boldsymbol{\theta}) + \beta J_{\text{smooth}}(\boldsymbol{\theta}) + \gamma J_{\text{energy}}(\boldsymbol{\theta})
\end{align}
where $\beta \ge 0$ controls heading regularization, and $\gamma \ge 0$ is the energy weight, with $\gamma=0$ disabling the energy term. Both terms are differentiable, so their gradients add to the backlog gradient:
\begin{align}
    \nabla_{\boldsymbol{\theta}} J_{\text{total}} = \nabla_{\boldsymbol{\theta}} J + \beta \nabla_{\boldsymbol{\theta}} J_{\text{smooth}} + \gamma \nabla_{\boldsymbol{\theta}} J_{\text{energy}}
\end{align}
For an interior slot $k$, the smoothness contribution to the heading gradient is
\begin{align}
    \frac{\partial J_{\text{smooth}}}{\partial \theta[k]}
    =\frac{2}{T-1}\big(2\theta[k]-\theta[k-1]-\theta[k+1]\big),
\end{align}
with the corresponding one-sided expressions at the endpoints and zero direct derivative with respect to $v[k]$.

Backpropagation through a long horizon $T$ resembles a deep network with shared weights, so gradient growth remains a practical concern. Proposition~\ref{prop:bounded} shows that the special block structure rules out geometric growth under stated boundedness assumptions, although polynomial growth remains possible. Norm clipping is therefore applied as a numerical safeguard. After computing $\mathbf{g} = \nabla_{\boldsymbol{\theta}} J_{\text{total}}$, we use
\begin{align}
    \hat{\mathbf{g}} = \min\!\left(1, \frac{C}{\|\mathbf{g}\|}\right) \mathbf{g},
    \label{eq:clip}
\end{align}
which preserves the descent direction while bounding the step magnitude.

\subsection{Convergence Analysis and Stability}
We distinguish three issues: correctness of a realized-rollout derivative, horizon-wise adjoint growth, and optimization of the resulting nonconvex objective. The statements below apply to the differentiable regions of the clipped-backlog dynamics; at a kink, AD returns a valid implementation-selected subgradient.

\begin{proposition}[Exact conditional pathwise gradient]\label{prop:zerovar}
Let $\boldsymbol{\epsilon}$ collect the initial-state and channel-noise variables used in a rollout. For fixed $\boldsymbol{\epsilon}$, reverse-mode AD returns the exact derivative $\nabla_{\boldsymbol{\theta}}J(\boldsymbol{\theta};\boldsymbol{\epsilon})$. If differentiation and expectation can be interchanged, then
\begin{align}
 \nabla_{\boldsymbol{\theta}}\mathbb{E}_{\boldsymbol{\epsilon}}[J]
 =\mathbb{E}_{\boldsymbol{\epsilon}}
 [\nabla_{\boldsymbol{\theta}}J(\boldsymbol{\theta};\boldsymbol{\epsilon})].
\end{align}
\end{proposition}
\begin{IEEEproof}
Conditioned on $\boldsymbol{\epsilon}$, the rollout is a finite composition of differentiable operations, and reverse-mode AD is the chain rule applied to that composition. The expectation identity follows from the assumed interchange. Although the conditional derivative has no numerical-estimation error, stochastic-rollout gradients generally have nonzero sampling variance.
\end{IEEEproof}

\begin{proposition}[No geometric growth from the physical-state Jacobian]\label{prop:bounded}
Suppose the active-set derivatives and the rate-position Jacobians are uniformly bounded by $c_R$. Then products of the open-loop physical-state Jacobians satisfy
$\|\mathbf{A}[T-1]\cdots\mathbf{A}[t]\|=\mathcal{O}(T-t)$, and the corresponding open-loop co-state is at most $\mathcal{O}((T-t)^2)$.
\end{proposition}
\begin{IEEEproof}
The Jacobian has block form
\begin{align}
\mathbf{A}[t]=
\begin{bmatrix}\mathbf{I}_2&\mathbf{0}\\
\mathbf{C}[t]&\mathbf{D}[t]\end{bmatrix},
\end{align}
where $\|\mathbf{C}[t]\|\le c_R$ and $\mathbf{D}[t]$ is diagonal with entries in $\{0,1\}$ away from clipping kinks. Multiplying these block-triangular matrices accumulates the lower-left block additively, giving a product norm linear in the number of factors rather than exponential. Summing the propagated bounded stage-cost derivatives in the adjoint recursion adds one further factor of the remaining horizon.
\end{IEEEproof}

This proposition excludes a geometric explosion caused by the stated physical dynamics; it does not claim a horizon-independent gradient bound, and the closed-loop policy Jacobian can introduce additional amplification. Gradient clipping is therefore a practical safeguard rather than a consequence of the proposition.

\begin{theorem}[Idealized full-gradient descent]\label{thm:convergence}
Assume $J$ is lower bounded and has an $L_g$-Lipschitz gradient. The idealized update
$\boldsymbol{\theta}_{k+1}=\boldsymbol{\theta}_k-\zeta\nabla J(\boldsymbol{\theta}_k)$
with $\zeta\in(0,1/L_g]$ satisfies
\begin{align}
J(\boldsymbol{\theta}_{k+1})\le J(\boldsymbol{\theta}_k)
-\frac{\zeta}{2}\|\nabla J(\boldsymbol{\theta}_k)\|^2,
\end{align}
and
\begin{align}
\min_{0\le k\le K_{\mathrm{it}}}\|\nabla J(\boldsymbol{\theta}_k)\|^2
\le\frac{2(J(\boldsymbol{\theta}_0)-J^\star)}
{\zeta(K_{\mathrm{it}}+1)}.
\end{align}
\end{theorem}
\begin{IEEEproof}
The result follows by applying the descent lemma for $L_g$-smooth functions \cite{nocedal2006numerical} and telescoping the per-iteration decrease.
\end{IEEEproof}

For a single exact-gradient step, clipping in \eqref{eq:clip} preserves the descent direction and the descent-lemma decrease whenever its effective step does not exceed $1/L_g$. The stationary-point rate in Theorem~\ref{thm:convergence}, however, is stated only for fixed-step full-gradient descent. The experiments use AdamW and may resample channel noise, so the theorem does not constitute a convergence guarantee for the implemented optimizer. Likewise, a stationary neural-policy parameter vector need not satisfy every action-wise PMP condition because parameters are shared across states and time.

\subsection{Comparison with Reinforcement Learning}
L4V and reinforcement learning address the same sequential decision problem but use different information. Policy-gradient RL estimates a score-function gradient from sampled returns and does not require derivatives of the environment. L4V instead differentiates a known simulator: one forward--backward pass gives the exact derivative of that realized rollout, while batches of stochastic rollouts are still needed to estimate the expected-objective gradient. The comparison therefore quantifies the benefit and cost of model-derivative access; it is not evidence that L4V has the model-free robustness of DRL. The adjoint provides dense temporal credit assignment through the dynamics, whereas a misspecified differentiable model can bias the resulting policy.
\begin{figure*}[!t]
\captionsetup{font={small}, skip=4pt}
    \centering
    \subfigure[Area: collection rate.]
    {
       \includegraphics[width=0.42\linewidth]{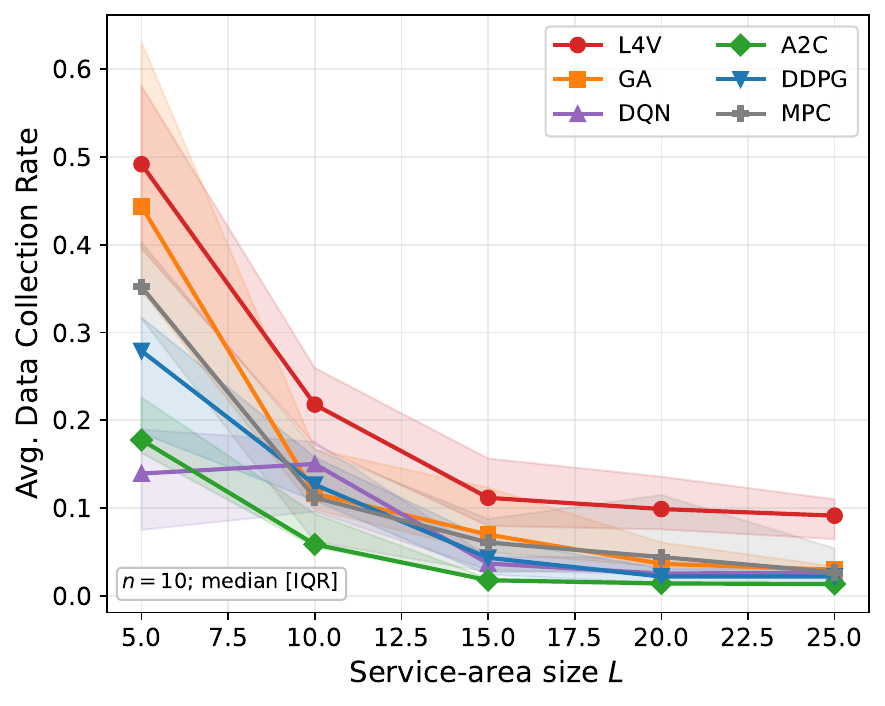}
       \label{fig:area_rate}
    }
    \hfil
    \subfigure[Area: completion time.]
    {
       \includegraphics[width=0.42\linewidth]{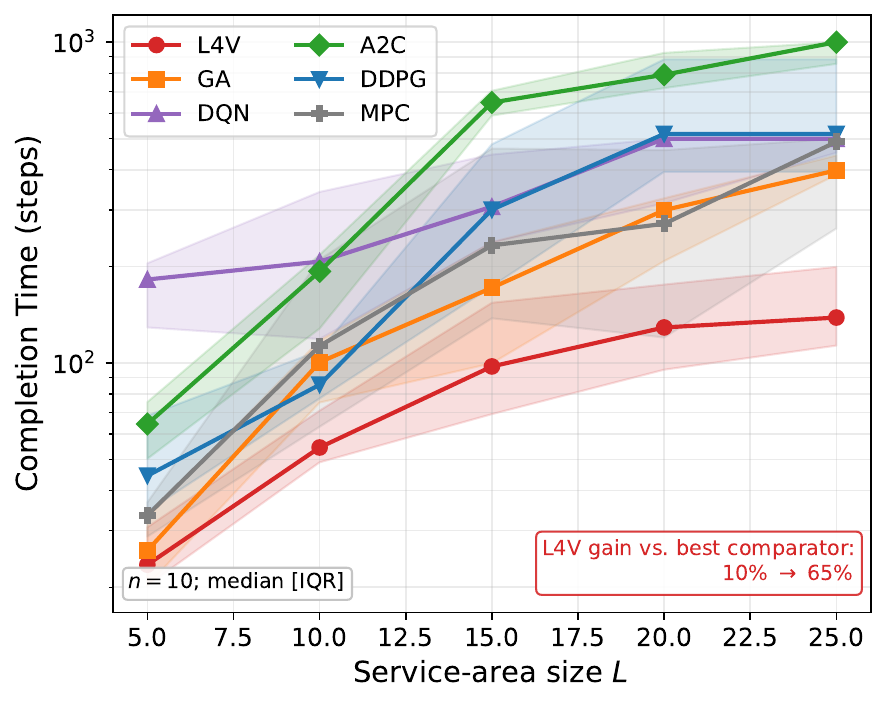}
       \label{fig:area_step}
    }
    \\
    \subfigure[Noise: collection rate.]
    {
       \includegraphics[width=0.42\linewidth]{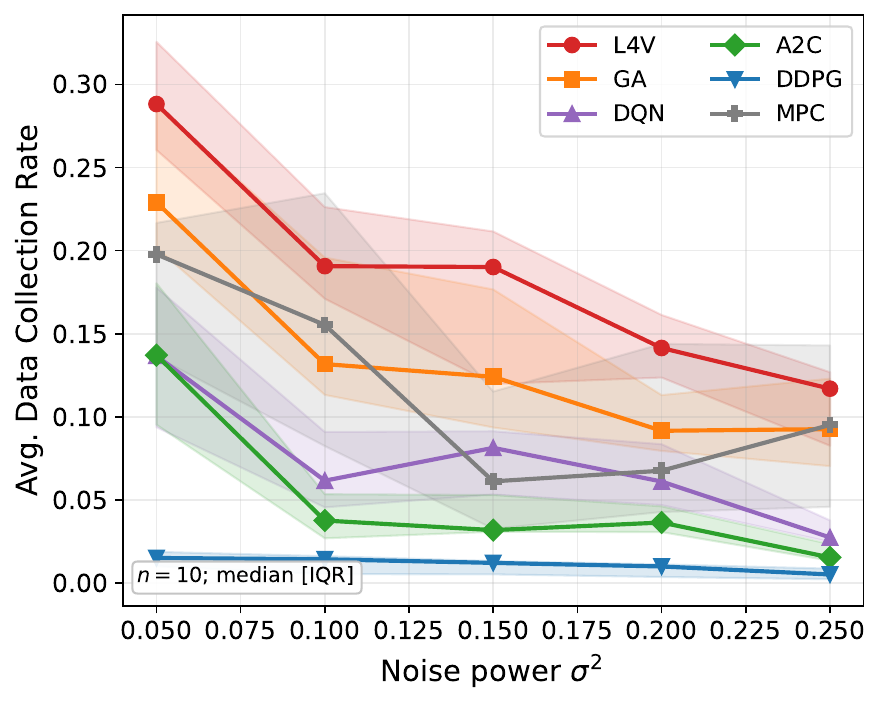}
       \label{fig:sigma_rate}
    }
    \hfil
    \subfigure[Noise: completion time.]
    {
       \includegraphics[width=0.42\linewidth]{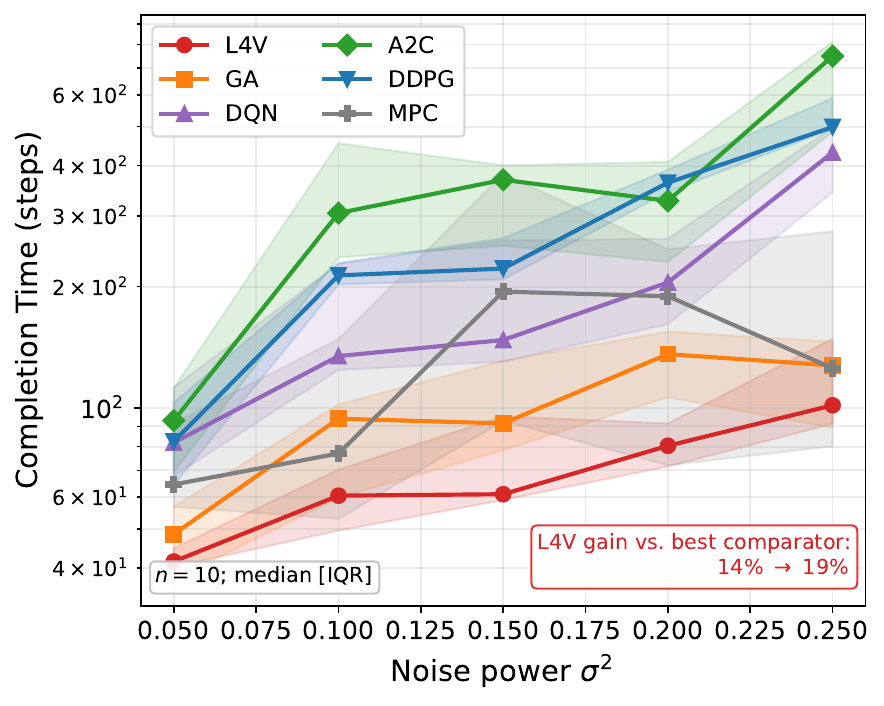}
       \label{fig:sigma_step}
    }
    \vspace{-6pt}
    \caption{Performance versus service-area size $L$ (a, b) and noise power $\sigma^2$ (c, d). Medians over $10$ trials with inter-quartile bands; completion time on a log axis.}
    \vspace{-8pt}
    \label{fig-area}
\end{figure*}

\section{Simulation Results}\label{sec-6}
\subsection{Simulation Setup}
We report completion time, average data collection rate (delivered demand divided by mission length), and compute cost. Each sweep setting contains $10$ independent trials and the curves show the median and inter-quartile range (IQR). In each trial, $K$ users are uniformly placed in $[-L/2,L/2]^2$, normalized demands are sampled from $[2,4)$, and the AAV starts at the origin. Every learned method is trained and evaluated on that same trial layout; no curve measures transfer to unseen layouts. The default configuration is $K=4$, $L=10$, $\eta=1$, $\sigma^2=0.1$, $\tau=0.1$, maximum per-slot displacement $v_{\max}=0.2$, and smoothness weight $\beta=0.01$. We vary $K\in\{2,4,6,8,10\}$, $L\in\{5,10,15,20,25\}$, $\eta\in\{0.5,1,1.5,2.0,2.5\}$, and $\sigma^2\in\{0.05,0.1,0.15,0.2,0.25\}$ one at a time. L4V uses a 64--64--32 MLP on the AAV position, user-relative positions, and remaining-demand vector. The broad sweeps use the deterministic equal-split channel ($a_i[t]=1/N$, $\xi_i[t]=1$, $\nu=2$), whereas Section~\ref{sec-6f} evaluates learned allocation under $\kappa=0.15$ and $K_{\mathrm r}=6$. A mission succeeds below total backlog $10^{-3}K$ and is otherwise capped at $500$ slots; capped runs remain in all statistics. L4V uses AdamW with learning rate $10^{-3}$ and at most $2000$ updates. The allocation, energy, and reviewer-driven stress studies use paired scenario seeds and retain raw per-rollout trajectories and outcomes.

\subsection{Baseline Methods}\label{sec-6b}
We compare L4V against five implementations spanning evolutionary optimization, value-based and actor-critic reinforcement learning, and gradient-based optimal control.
\begin{itemize}
    \item \textbf{Genetic Algorithm (GA)} \cite{lambora2019genetic}: a derivative-free evolutionary method that maintains a population of candidate trajectories improved by selection, crossover, and mutation, representing meta-heuristic optimization that uses no gradient information.
    \item \textbf{Deep Q-Network (DQN)} \cite{mnih2015human}: a value-based reinforcement learning method; since it requires a discrete action space, we discretize the heading angle, which highlights the difficulty of applying discrete control to a smooth physical trajectory task.
    \item \textbf{Advantage Actor-Critic (A2C)} \cite{grondman2012survey}: an on-policy actor-critic algorithm that optimizes a stochastic policy via the advantage function, illustrating the impact of high-variance stochastic gradients against the deterministic adjoint gradient of L4V.
    \item \textbf{Deep Deterministic Policy Gradient (DDPG)} \cite{silver2014deterministic}: an off-policy actor-critic method for continuous control using replay buffers, serving as the primary reinforcement learning benchmark.
    \item \textbf{Differentiable Model Predictive Control (MPC)}: a receding-horizon shooting controller on the same differentiable simulator used by L4V. At each slot it optimizes an $H_p$-step sequence, executes the first action, and re-solves. This provides a gradient-based control reference with the same model access but a different compute profile.
\end{itemize}
For reproducibility, the learned methods share the same state encoding and 64--64--32 hidden widths, but their action spaces and budgets follow their implementations. DQN uses 256 discrete headings, Adam at $10^{-3}$, a 10\,000-transition replay buffer, batch size 64, discount $0.99$, $\epsilon$ from $1$ to $0.01$ with factor $0.995$, a target update every 10 iterations, and 2000 updates. A2C uses a categorical eight-heading policy, Adam at $10^{-3}$, discount $0.99$, entropy weight $0.01$, and 1000 updates. DDPG controls continuous speed and heading with actor/critic learning rates $10^{-3}$/$2\times10^{-3}$, replay size 10\,000, batch size 64, discount $0.99$, target factor $0.005$, exploration noise $0.2$ to $0.01$ with factor $0.995$, and 2000 updates. GA uses population 300, crossover 0.7, mutation 0.3, and at most 5000 generations. MPC uses horizon 15 and 25 inner Adam steps at learning rate $0.05$ per replanning solve. These differences prevent a strict equal-compute claim; the reported wall-clock and mission metrics expose the resulting trade-off.
\begin{figure*}[!t]
\captionsetup{font={small}, skip=4pt}
    \centering
    \subfigure[Density: collection rate.]
    {
       \includegraphics[width=0.42\linewidth]{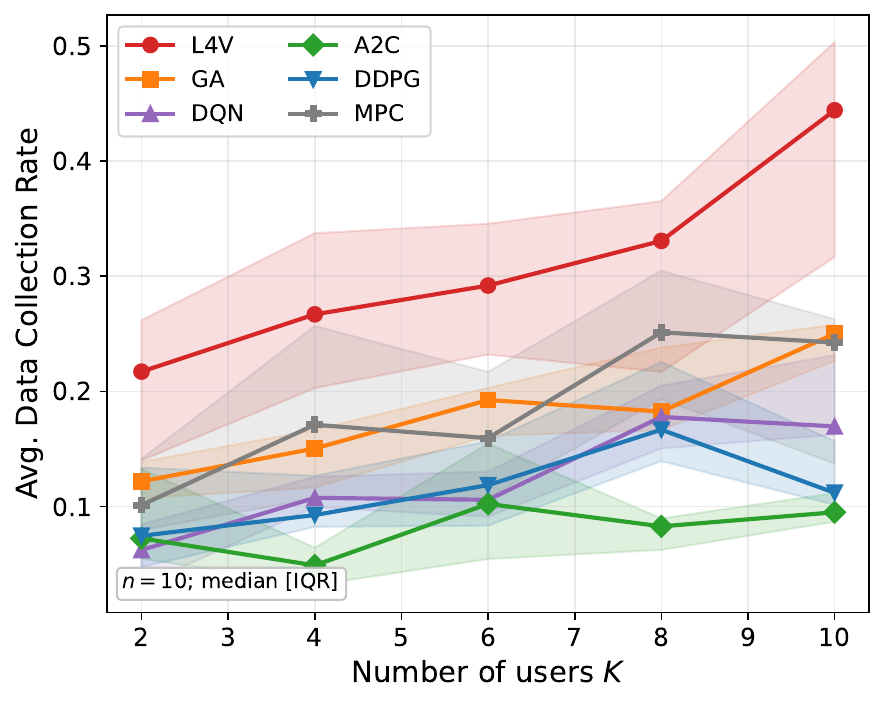}
       \label{fig:user_rate}
    }
    \hfil
    \subfigure[Density: completion time.]
    {
       \includegraphics[width=0.42\linewidth]{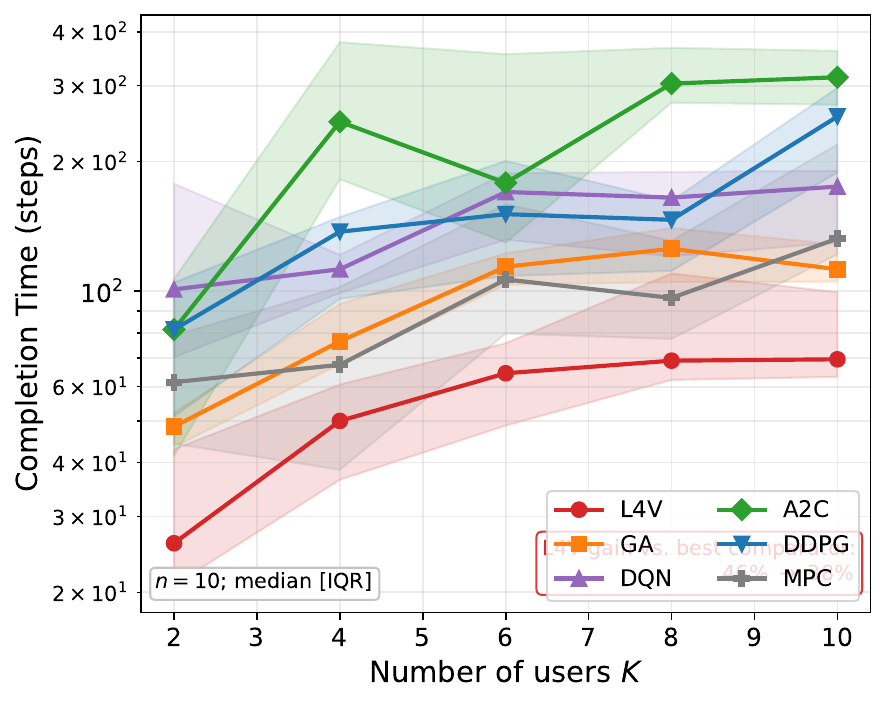}
       \label{fig:user_step}
    }
    \\
    \subfigure[Gain: collection rate.]
    {
       \includegraphics[width=0.42\linewidth]{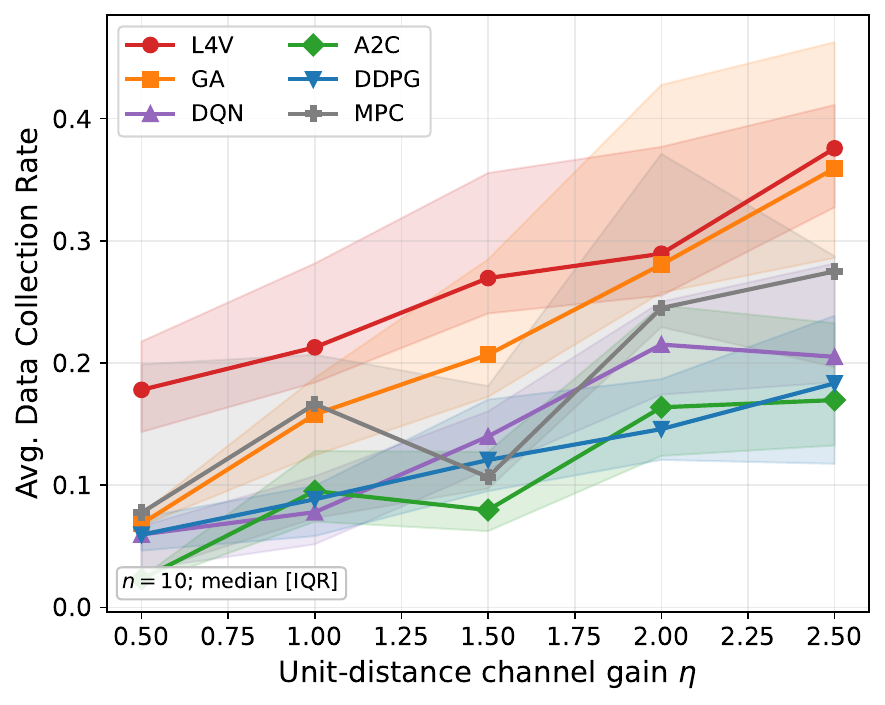}
       \label{fig:hunit_rate}
    }
    \hfil
    \subfigure[Gain: completion time.]
    {
       \includegraphics[width=0.42\linewidth]{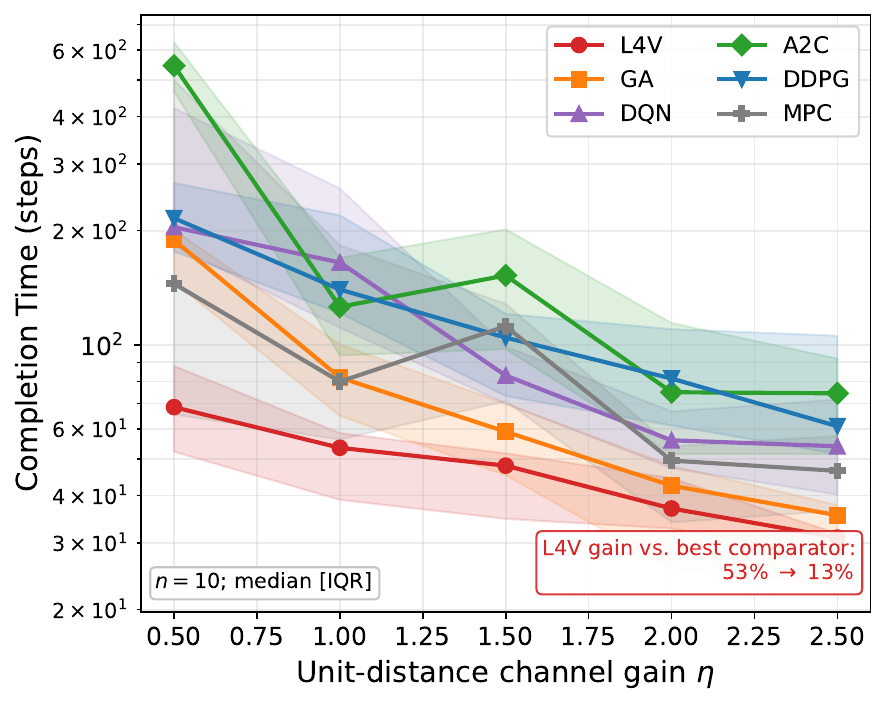}
       \label{fig:hunit_step}
    }
    \vspace{-6pt}
    \caption{Performance versus user density $K$ (a, b) and channel gain $\eta$ (c, d). Medians over $10$ trials with inter-quartile bands; completion time on a log axis.}
    \vspace{-8pt}
    \label{fig-userhunit}
\end{figure*}

\subsection{Impact of Service-Area Size}\label{sec-6c}
Fig.~\ref{fig-area} reports performance as the service-area side length $L$ grows from $5$ to $25$. L4V has the lowest median completion time at every tested scale, and its reduction relative to the strongest comparator grows from $10\%$ to $65\%$. A larger area increases typical travel distance and completion time; A2C and DDPG grow by roughly an order of magnitude over this range, whereas the L4V curve grows more slowly. Panel (a) reports the corresponding average collection rate and preserves the same empirical ordering.

\subsection{Impact of Noise Power}\label{sec-6d}
Fig.~\ref{fig-area}(c, d) varies the deterministic noise-power parameter $\sigma^2$ from $0.05$ to $0.25$. Higher noise lowers the achievable rate and lengthens every mission. L4V has the lowest median completion time throughout, with a $14\%$--$41\%$ reduction relative to the strongest comparator at each level. These sweeps do not sample channel fading, so they test sensitivity to link quality rather than stochastic robustness; the paired stochastic-channel study is reported separately in Section~\ref{sec-6f}.

\subsection{Impact of User Density}\label{sec-6e}
Fig.~\ref{fig-userhunit}(a, b) varies the number of users $K$ from $2$ to $10$. L4V has the lowest median completion time at every tested density, with reductions of $35\%$--$46\%$ relative to the strongest comparator. Its curve grows more slowly than those of the implemented baselines over this range. Panel (a) shows that the average collection rate rises with density for all methods, while preserving the empirical ordering.

\begin{table}[!t]
\centering
\caption{Scenario-specific optimization and execution costs at the default configuration. Neural execution is the warmed median rollout cost scaled to the observed median mission length; GA and MPC report median per-instance solve time. Held-out transfer is evaluated separately in Fig.~\ref{fig:stress_zeroshot}.}
\label{tab:compute}
\resizebox{0.86\linewidth}{!}{
\begin{tabular}{@{}l|c|c@{}}
\hline
Method & Scenario optimization & Execution/solve \\
\hline
\rowcolor{red!8}\textbf{L4V} & \textbf{293~s} & \textbf{53~ms} \\
GA   & -- & 2369~s \\
DQN  & 887~s & 118~ms \\
A2C  & 2475~s & 131~ms \\
DDPG & 3506~s & 111~ms \\
MPC  & -- & 64~s \\
\hline
\end{tabular}}
\end{table}

Table~\ref{tab:compute} separates policy optimization from online execution. For neural policies, we warm the implementation, measure the median rollout cost over $30$ runs, and scale it to the observed median mission length; GA and MPC report per-instance solve time. Once optimized, L4V's $53$-ms execution estimate is approximately $1.2\times10^3$ times lower than the implemented MPC solve and $4.5\times10^4$ times lower than GA. These default-configuration measurements are not an equal-compute comparison. Section~\ref{sec-6j} separately tests whether distributional pretraining makes the forward-only execution mode available on held-out layouts.

\begin{figure*}[!t]
\captionsetup{font={small}, skip=6pt}
    \centering
    \subfigure[Completion time.]
    {
       \includegraphics[width=0.42\linewidth]{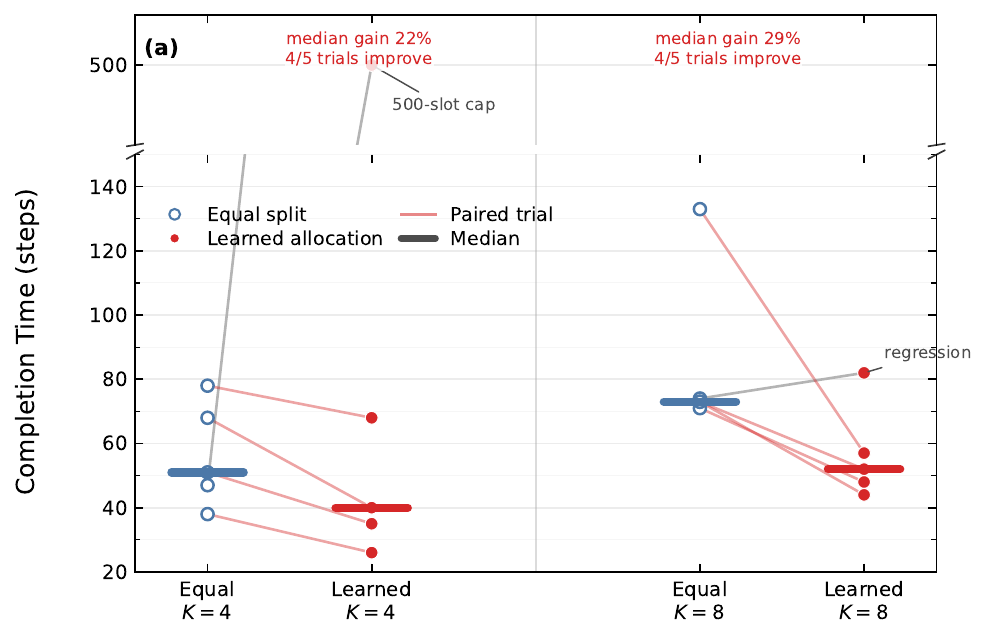}
       \label{fig:alloc_step}
    }
    \hfil
    \subfigure[Total flight energy.]
    {
       \includegraphics[width=0.42\linewidth]{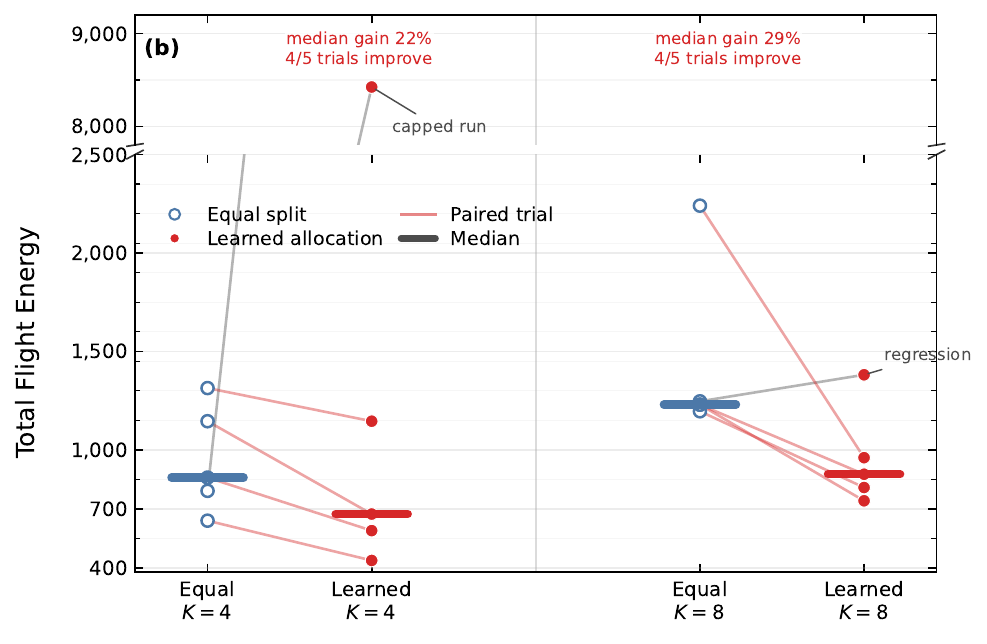}
       \label{fig:alloc_energy}
    }
    \\
    \subfigure[Energy--time relation.]
    {
       \includegraphics[width=0.42\linewidth]{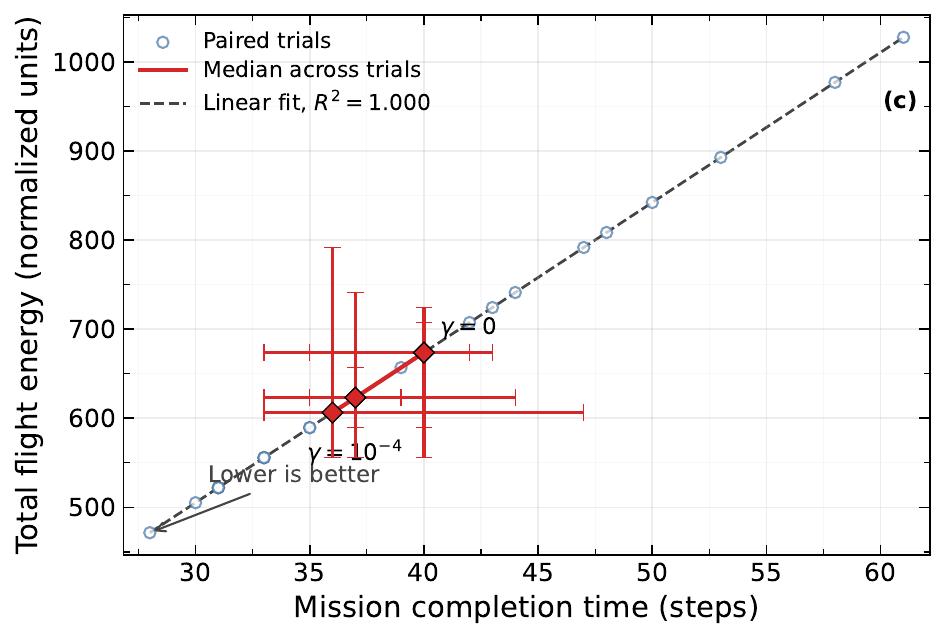}
       \label{fig:energy_pareto}
    }
    \hfil
    \subfigure[Paired change from $\gamma=0$.]
    {
       \includegraphics[width=0.42\linewidth]{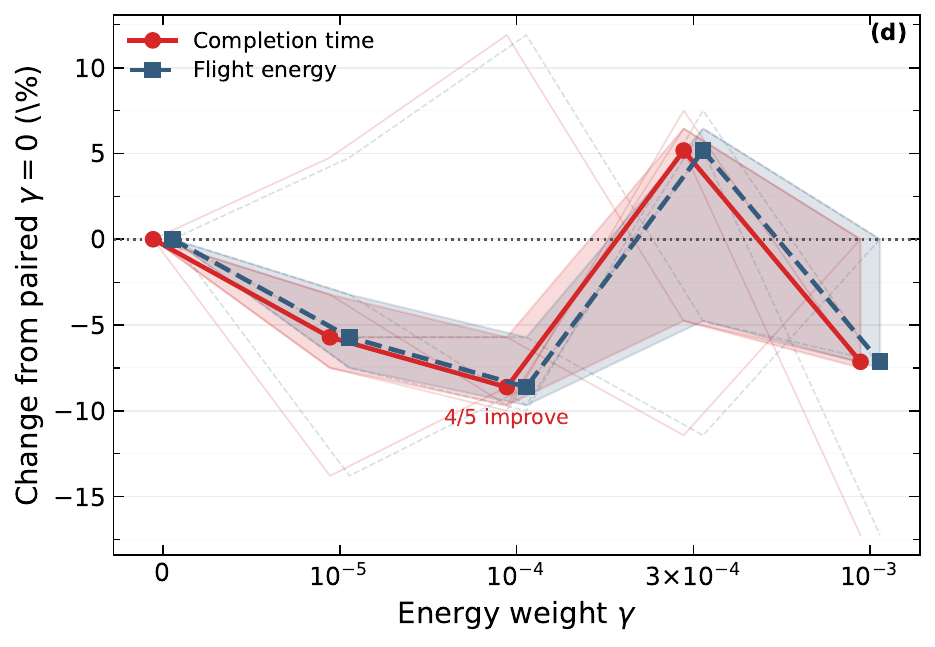}
       \label{fig:energy_normalized}
    }
    \vspace{-6pt}
    \caption{Allocation ablation and energy-weight sensitivity. Panels (a, b) compare learned allocation with equal split over five paired instances; thin segments preserve each pairing and thick horizontal bars mark medians. The broken axes omit empty ranges while retaining and labeling the capped run. Panel (c) shows every energy-sweep trial with median/IQR summaries and a linear diagnostic. Panel (d) reports paired percentage changes from $\gamma=0$ with inter-quartile bands. Lower is better throughout, and capped or regressing trials are retained.}
    \vspace{-10pt}
    \label{fig-alloc}
\end{figure*}

\begin{figure}[!t]
    \centering
    \includegraphics[width=0.78\linewidth]{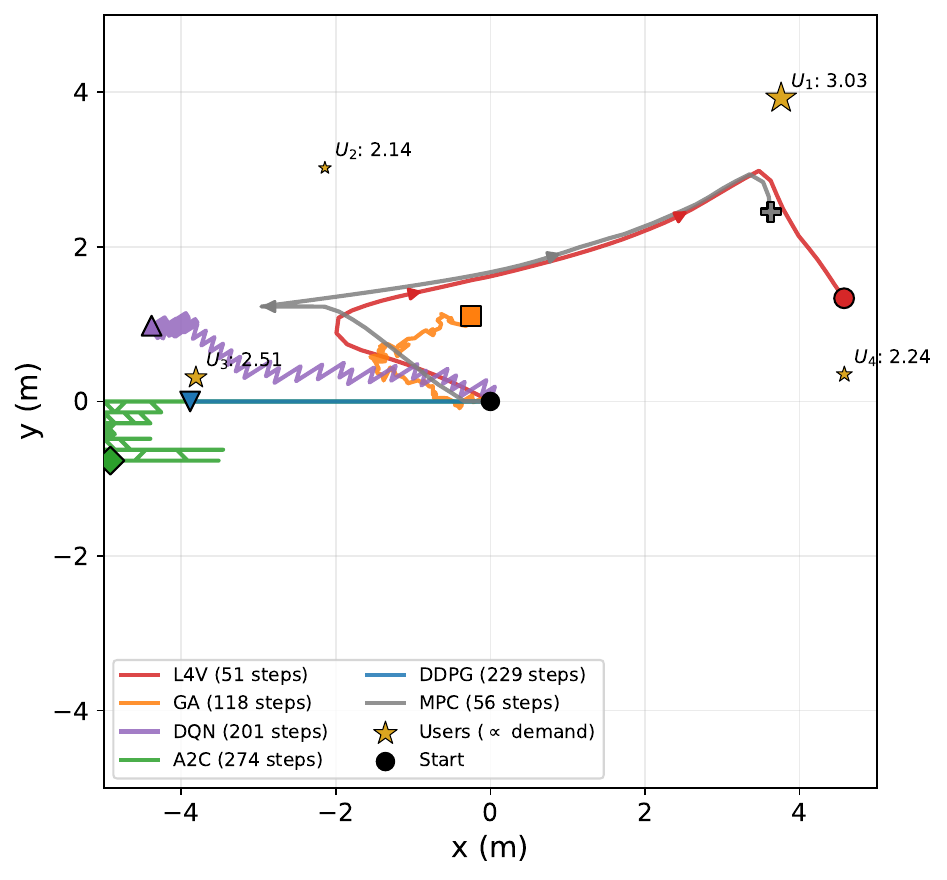}
    \caption{Executed AAV paths on one shared $K=4$ instance; marker size scales with task demand and each path is labeled with its mission length.}
    \label{fig-traj}
\end{figure}

\begin{figure*}[!t]
\captionsetup{font={small}, skip=4pt}
    \centering
    \subfigure[Channel-generator mismatch.]{
       \includegraphics[width=0.42\linewidth]{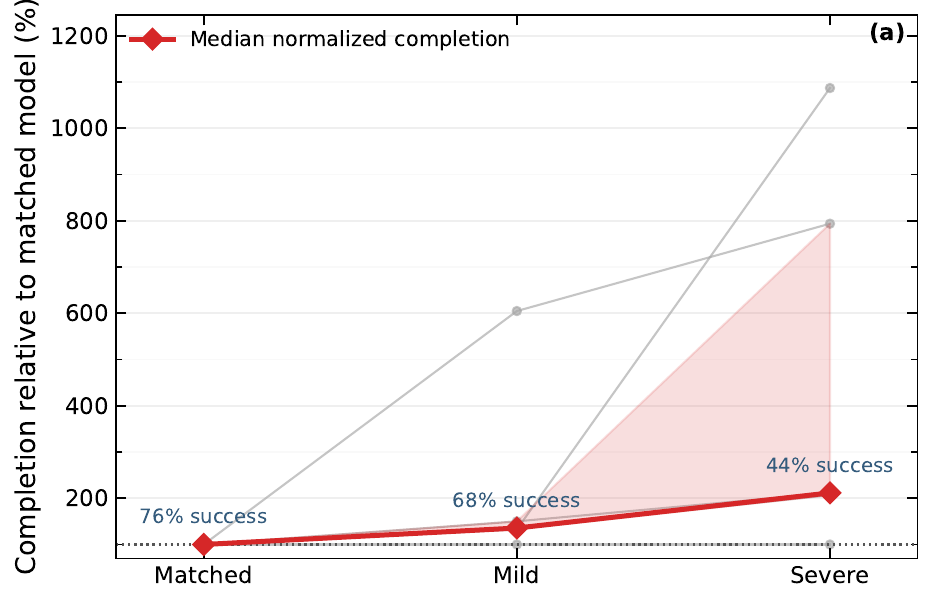}
       \label{fig:stress_mismatch}}
    \hfil
    \subfigure[Partial observations.]{
       \includegraphics[width=0.42\linewidth]{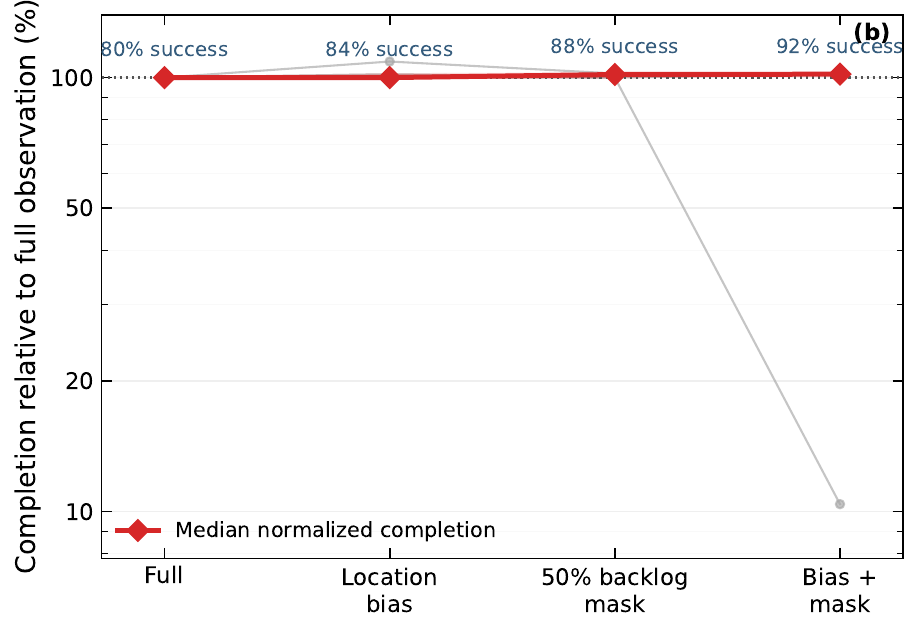}
       \label{fig:stress_partial}}
    \\
    \subfigure[Fixed-resource two-AAV extension.]{
       \includegraphics[width=0.42\linewidth]{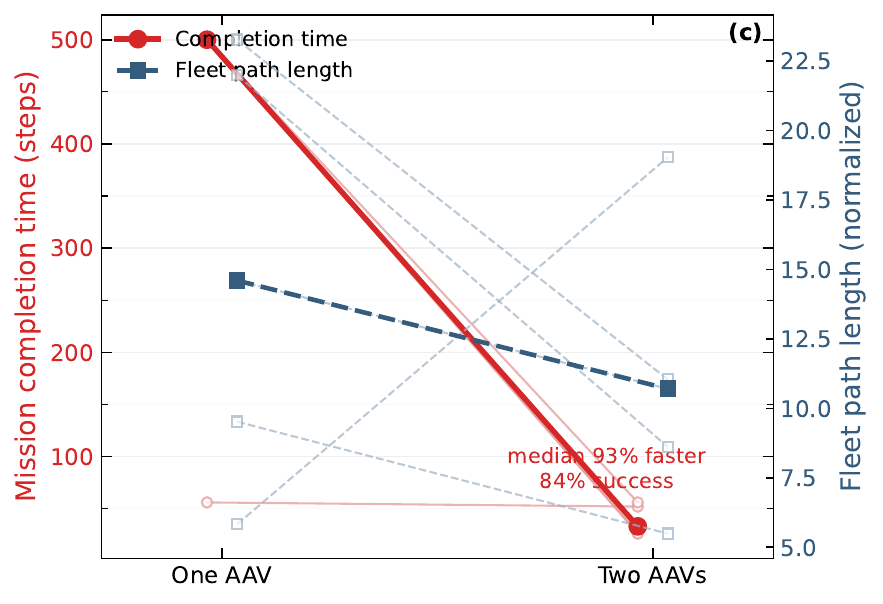}
       \label{fig:stress_multi}}
    \hfil
    \subfigure[Circular no-fly region.]{
       \includegraphics[width=0.42\linewidth]{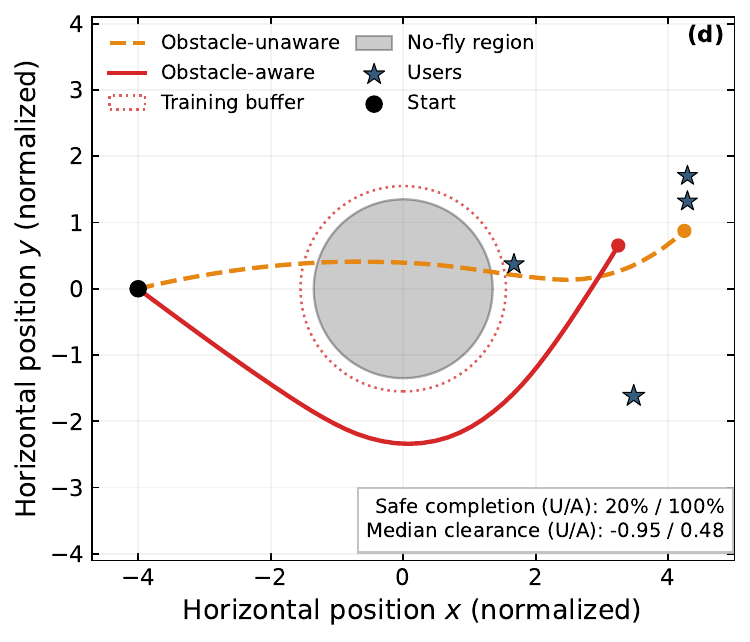}
       \label{fig:stress_obstacle}}
    \\
    \subfigure[Zero-shot held-out layouts.]{
       \includegraphics[width=0.42\linewidth]{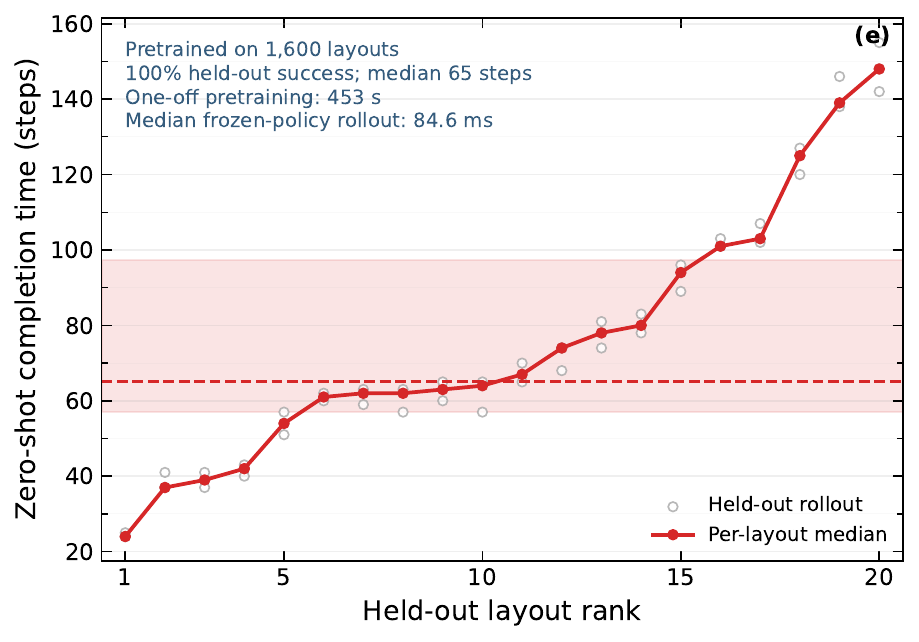}
       \label{fig:stress_zeroshot}}
    \hfil
    \subfigure[Heading-smoothness sensitivity.]{
       \includegraphics[width=0.42\linewidth]{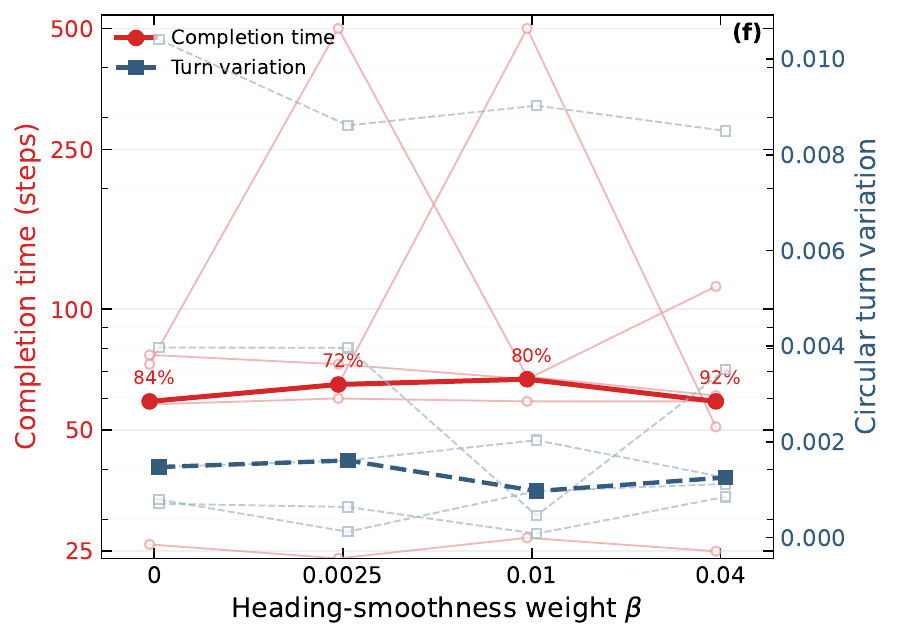}
       \label{fig:stress_smooth}}
    \vspace{-6pt}
    \caption{Reviewer-driven stress tests. Panels (a, b) keep the trained policy fixed and pair the random draws across deployment conditions. The nominal channel is $(\kappa,K_{\rm r},\nu)=(0.15,6,2)$; mild and severe shifts are $(0.25,3,2.2)$ and $(0.40,0,2.5)$. Observation tests add a fixed location bias of standard deviation $0.30$, independently mask each current backlog update with probability $0.50$ while holding its last estimate, or combine both. Panel (c) retains one global radio budget as the fleet grows from one to two AAVs. Panel (d) applies hard collision checks after training with or without obstacle geometry and a soft clearance barrier. Panel (e) freezes one distributionally pretrained policy on disjoint layouts. Panel (f) reports completion and circular turn variation. Thin marks retain trial-level outcomes; thick marks and bands show medians and IQRs; capped runs and success rates are retained. Logarithmic completion-time ordinates in (b) and (f) preserve extreme trials without compressing the median region.}
    \vspace{-8pt}
    \label{fig-stress}
\end{figure*}

\subsection{Impact of Channel Gain}\label{sec-6h}
Fig.~\ref{fig-userhunit}(c, d) varies the unit-distance gain $\eta$ from $0.5$ to $2.5$. Every completion-time curve decreases monotonically. L4V has the lowest median at every tested gain; its reduction relative to the strongest comparator is $64\%$ at $\eta=0.5$ and narrows to $13\%$ at $\eta=2.5$. The average collection rate in panel (c) rises with $\eta$ for every method.
\subsection{Extended Model: Joint Trajectory and Allocation Learning under a Stochastic Channel}\label{sec-6f}
The preceding sweeps hold the equal-split OFDMA channel fixed. We now compare the extended allocation policy with an otherwise identical equal-split policy under reparameterized shadowing and Rician fading. The five trials are paired by user layout, demand, and random seed, so the contrast isolates allocation within this in-distribution stochastic generator. Fig.~\ref{fig-alloc} shows median completion-time reductions of $22\%$ at $K=4$ and $29\%$ at $K=8$, with improvements in four of five pairs at each size. The full paired view also exposes a $K=4$ learned-policy run that reaches the $500$-slot cap and one $K=8$ regression. Thus allocation is useful on most tested instances but is not uniformly beneficial. Because no mismatched generator or partial observation is tested, this ablation supports the value of shared-resource control rather than a broad robustness claim.
\subsection{Completion-Time-Versus-Energy Tradeoff}\label{sec-6g}
Figures~\ref{fig:energy_pareto} and~\ref{fig:energy_normalized} report all five paired $K=4$ instances for $\gamma \in \{0,10^{-5},10^{-4},3\times10^{-4},10^{-3}\}$. At $\gamma=10^{-4}$, four of five instances improve relative to $\gamma=0$, and the medians change from $40$ to $36$ slots and from $674$ to $606$ energy units. However, the raw trials lie almost exactly on one line ($R^2=1.000$), and energy per slot is nearly constant in this normalized low-speed regime. The observed energy reduction is therefore explained primarily by shorter missions rather than an independently resolved speed--energy Pareto frontier. The remaining instance regresses and the trend is not monotonic in $\gamma$. These results demonstrate differentiable energy-objective integration while delimiting what the data support; a calibrated physical-speed regime remains future work.
\subsection{Learned Trajectory Visualization}\label{sec-6i}
Figure~\ref{fig-traj} overlays all methods on one shared instance and annotates user demand, travel direction, and mission length. L4V completes in $51$ slots and MPC in $56$ slots along visually similar tours. DDPG remains near one user for $229$ slots; DQN and A2C oscillate and take $201$ and $274$ slots; GA takes $118$ slots. This example illustrates one source of the aggregate gap, but it is neither representative by construction nor an optimality certificate.

\subsection{Robustness, Structural Extensions, and Zero-Shot Transfer}\label{sec-6j}
Figure~\ref{fig-stress} separates six questions that the preceding parameter sweeps do not answer. For panels (a)--(d) and (f), five scenario seeds are trained independently and evaluated under five paired channel rollouts per condition. Channel mismatch, location bias, and partial backlog updates are applied only at deployment, leaving policy parameters frozen. The two-AAV controller uses joint state information, a differentiable inter-AAV separation penalty, and one global softmax over all AAV--user pairs; hence its total radio budget equals that of the single-AAV case. In the obstacle test, all users lie beyond a circular no-fly region relative to the start. The unaware policy receives no obstacle coordinates, whereas the aware policy receives the circle and a buffered soft clearance barrier during training; evaluation uses a hard collision indicator.

Panels (a) and (b) quantify bounded departures from training information rather than arbitrary robustness. Across $25$ rollouts per condition, the channel-matched, mild-shift, and severe-shift cases have median completion times of $63$, $117$, and the $500$-slot cap, with success rates of $76\%$, $68\%$, and $44\%$, respectively. The monotone degradation shows that L4V remains sensitive to its assumed generator. Under full observation, location bias, $50\%$ backlog-update masking, and their combination, the medians are $57$, $56$, $49$, and $49$ slots and success rates are $80\%$, $84\%$, $88\%$, and $92\%$. The held-last-value backlog estimate is conservative and does not cause degradation in these paired tests, but this result neither covers missing location identities nor establishes general partially observable control. In panel (c), the one- and two-AAV cases attain median completion times of $500$ and $50$ slots and success rates of $40\%$ and $84\%$, respectively, while the two-AAV median total fleet distance is also lower ($10.7$ versus $14.6$ normalized units). This demonstrates that added mobility can help even when the total radio budget is fixed, but it is a centralized structural extension rather than a decentralized multi-agent solution.

Panel (d) uses a $0.20$-unit training buffer around the hard no-fly boundary and a barrier weight of $250$. The obstacle-unaware and obstacle-aware policies have median completion times of $500$ and $46$ slots, safe-completion rates of $20\%$ and $100\%$, collision rates of $80\%$ and $0\%$, and median minimum clearances of $-0.95$ and $0.48$, respectively. The test establishes differentiable avoidance for one known circular obstacle; it does not cover unknown maps, moving obstacles, or non-differentiable collision dynamics.

Panel (e) evaluates a distinct distributional-pretraining protocol. A single 64--64--32 policy is pretrained across $1{,}600$ independently generated layouts and then frozen. Test layouts use $20$ disjoint seeds with three channel rollouts each; no gradient update, fine-tuning, or online trajectory optimization is permitted. The measured one-off pretraining time is $453$ s. All $60$ held-out rollouts complete, with a median of $65$ slots, an inter-quartile range of $57$--$97$ slots, and a median policy-rollout time of $85$ ms. This is zero-shot generalization within the sampled layout and channel distribution, not a claim of arbitrary out-of-distribution robustness. It establishes that large-scale world-model-guided pretraining can amortize trajectory optimization into direct deployment on new layouts; the broad sweeps remain per-layout optimization comparisons.

Panel (f) sweeps $\beta\in\{0,0.0025,0.01,0.04\}$. The median completion times are $60$, $67$, $67$, and $57$ slots, with success rates of $84\%$, $72\%$, $80\%$, and $92\%$; median circular turn variation is $1.43$, $1.31$, $0.97$, and $1.18$ in units of $10^{-3}$. Relative to $\beta=0$, the default $\beta=0.01$ reduces turn variation by $32\%$ but increases the median mission by seven slots, whereas $\beta=0.04$ is faster but less smooth than $\beta=0.01$. The response is nonmonotonic, so $\beta$ remains a task-dependent regularization choice.

\FloatBarrier
\section{Conclusion}\label{sec-7}
This paper has developed L4V, which treats coupled AAV motion, stochastic radio propagation, shared allocation, and residual demand as a structured differentiable wireless world model and converts its derivatives into a reusable policy. It has linked a cumulative-backlog surrogate to open- and closed-loop adjoints, scoped the resulting gradient and convergence statements, and tested bounded departures from the training model. Across the implemented baselines, L4V reduces median mission time by up to $65\%$, while one policy pretrained on $1{,}600$ layouts completes all $60$ held-out frozen-policy rollouts; paired stress tests retain the failures under severe channel mismatch rather than implying model-free robustness or global optimality. Future work will address moving users, trace-driven site-specific channels and maps, discrete scheduling and interference, decentralized fleet coordination, calibrated propulsion, and hardware experiments.

\balance
\bibliography{ref}
\bibliographystyle{IEEEtran}
\ifCLASSOPTIONcaptionsoff
  \newpage
\fi
\end{document}